\documentclass[11pt]{scrartcl}
\usepackage{geometry} 
\usepackage[backref]{hyperref}

\hypersetup{
    unicode=true,
    pdftitle={Explicit relation between all lower bound techniques for quantum query complexity}, 
    pdfauthor={Lo\"ick Magnin, J\'er\'emie Roland},
    colorlinks=true,       
    linkcolor=red!50!black,          
    citecolor=blue!50!black,        
    urlcolor=blue!50!black          
}
\usepackage{amsmath}
\usepackage{amsfonts}
\usepackage{amssymb}
\usepackage{amsthm}
\usepackage[utf8]{inputenc}
\usepackage{pifont}
\usepackage[small,bf,sf]{caption}
\usepackage{aliascnt}
\usepackage{thm-restate}
\usepackage{microtype}

\makeatletter
\def\tagform@#1{\maketag@@@{\ignorespaces#1\unskip\@@italiccorr}}
\let\orgtheequation\theequation
\def\theequation{(\orgtheequation)}
\makeatother

\usepackage{pgf}
\usepackage{tikz}
\usetikzlibrary{shapes.geometric}
\usetikzlibrary{shapes.misc}
\usetikzlibrary{shapes.symbols}
\usetikzlibrary{decorations.pathmorphing}
\usetikzlibrary{decorations.pathreplacing}
\usetikzlibrary{positioning,backgrounds,shadows}
\newcommand{\pbsc}[1]{\textsc{#1}}

\newcommand{\ED}{\pbsc{Element Distinctness}}
\newcommand{\collision}{\pbsc{Collision}}
\newcommand{\ksum}{$k$-\pbsc{Sum}}
\newcommand{\eps}{\varepsilon}
\DeclareMathOperator{\adeg}{\widetilde{deg}}

\newcommand{\adv}{\mathrm{ADV}}
\newcommand{\ncqq}{Q^{\textrm{nc}}}
\newcommand{\cqq}{Q}
\newcommand{\HF}{\FF_H}
\newcommand{\Mtarget}{M}
\DeclareMathOperator{\xpoly}{xpoly}
\DeclareMathOperator{\madv}{MADV}

\newcommand{\ket}[1]{| #1 \rangle}
\newcommand{\bra}[1]{\langle #1 |}
\newcommand{\braket}[2]{\langle #1 | #2 \rangle}
\newcommand{\ketbra}[2]{| #1 \rangle\!\langle #2 |}
\newcommand{\proj}[1]{| #1 \rangle\!\langle #1 |}
\newcommand{\norm}[1]{\left\| #1 \right\|}

\newcommand{\abs}[1]{\left| #1 \right|}
\newcommand{\inv}[1]{\frac{1}{ #1 }}

\DeclareMathOperator{\tr}{tr}
\DeclareMathOperator{\Span}{span}

\renewcommand{\C}{\mathbb{C}}  \newcommand{\I}{\mathbb{I}} \newcommand{\J}{\mathbb{J}} \newcommand{\R}{\mathbb{R}}

\newcommand{\FF}{\mathcal{F}}      
\newcommand{\HH}{\mathcal{H}}
\renewcommand{\SS}{\mathcal{S}}

\newtheoremstyle{sansthesis}{5pt}{5pt}{\itshape}{}{\bfseries\sffamily}{}{.7em}{} 
\theoremstyle{sansthesis}

\newtheorem{thm}{Theorem}

\newaliascnt{definition}{thm}
\newtheorem{definition}[definition]{Definition}
\aliascntresetthe{definition}

\newcommand{\dt}[1]{\textbf{#1}}

\newaliascnt{lem}{thm}
\newtheorem{lem}[lem]{Lemma}
\aliascntresetthe{lem}

\newaliascnt{claim}{thm}
\newtheorem{claim}[claim]{Claim}
\aliascntresetthe{claim}

\newaliascnt{fact}{thm}

\aliascntresetthe{fact}

\title{Explicit relation between all lower bound techniques for quantum query complexity}

\author{
	Loïck Magnin\thanks{NEC Laboratories America and Centre for Quantum Technologies, National University of Singapore; \href{mailto:loick@locc.la}{loick@locc.la}}
	 \and
	Jérémie Roland\thanks{NEC Laboratories America and Universit\'e Libre de Bruxelles; \href{mailto:jroland@ulb.ac.be}{jroland@ulb.ac.be}}
}
\date{}

\begin{document}

\maketitle 
\begin{abstract}
The polynomial method and the adversary method are the two main techniques to prove lower bounds on quantum query complexity, and they have so far been considered as unrelated approaches. Here, we show an explicit reduction from the polynomial method to the multiplicative adversary method.  The proof goes by extending the polynomial method from Boolean functions to quantum state generation problems. In the process, the bound is even strengthened. We then show that this extended polynomial method is a special case of the multiplicative adversary method with an adversary matrix that is independent of the function. This new result therefore provides insight on the reason why in some cases the adversary method is stronger than the polynomial method. It also reveals a clear picture of the relation between the different lower bound techniques, as it implies that all known techniques reduce to the multiplicative adversary method.
\end{abstract}


\section{Introduction}
\paragraph{Polynomial and adversary methods.} There are two main techniques to prove lower bounds on quantum query complexity: the polynomial method~\cite{BBC+01,KSW07,She11}, based on bounding the degree of the function seen as a polynomial, and adversary methods~\cite{BBB+97,Amb02,BS04, LM08,HNS08}, based on bounding the change in a progress function from one query to the next. In its original form~\cite{Amb02}, the adversary method bounds the additive change in the progress function, hence we will call it \emph{additive}, and the progress function is based on a matrix assigning positive weights to pairs of inputs. The polynomial method and this original adversary method are not comparable. Indeed, the original adversary method is limited by the ``certificate complexity barrier" \cite{Zha05,SS06}, that is, for total functions, $\adv(f) \leq \sqrt{C_0(f)C_1(f)}$ where $C_b(f)$ denotes the certificate complexity of $f$ for $f(x)=b$. It means that the original adversary method cannot prove lower bounds better than $\Omega(N^{1/2})$ for \ED. However, Aaronson and Shi \cite{AS04}  were able to prove a $\Omega(N^{2/3})$ lower bound using the polynomial method. On the other hand it is known that the adversary method can sometimes give better lower bounds than the polynomial method, in~\cite{Amb06} Ambainis exhibits a function with polynomial degree $d$ and adversary bound $\Omega(d^{1.3})$.

H{\o}yer, Lee, and {\v{S}}palek have extended the additive adversary method by allowing negative weights in the matrix~\cite{HLS07}, and have shown that the corresponding bound, $\adv^\pm(f)$, breaks the certificate complexity barrier. For simplicity, we will from now on refer to $\adv^\pm(f)$ as the additive adversary bound, implicitly allowing negative weights.

Recently, a series of works~\cite{FGG07,ACR+10,RS08,Rei11,LMRSS11} culminated in showing that this bound is tight in the bounded-error case for any function. However, this fundamental result does not answer all the questions about quantum query complexity as it suffers from two limitations. First, in some cases it is necessary to prove bounds for very small success probabilities, a regime where $\adv^\pm(f)$ might not be tight. For this reason, while the optimality of the additive adversary bound implies that quantum query complexity satisfies a direct sum theorem, it cannot be used to prove a strong direct product theorem, which requires to prove nontrivial bounds for exponentially small success probabilities. Secondly, while the proof of optimality of $\adv^\pm(f)$ implies that if a lower bound on the bounded-error quantum query complexity of a function can be proved with any method, it can also be proved with $\adv^\pm(f)$, this reduction is not constructive. Concretely, there are still examples of lower bounds that can be proved using the polynomial method for which the optimal adversary matrix is unknown, a typical example being the $\collision$ problem~\cite{AS04}\footnote{Until very recently it was also the case for the $\ED$ problem, whose lower bound was proved by reduction to $\collision$, but a direct adversary lower bound has now been shown by Belovs~\cite{Bel12}, and later extended to the $\ksum$ problem by Belovs and \v{S}palek~\cite{BS13}.}. 

\paragraph{Multiplicative adversary method.}
The first limitation has been overcome thanks to the introduction of another adversary-type method. By formalizing an ad-hoc technique proposed by Ambainis, de Wolf and {\v{S}}palek~\cite{Amb05-ASdW06,ASdW06}, {\v{S}}palek designed a new lower bound method which he called the multiplicative adversary method~\cite{Spa08}, as the idea is to bound the multiplicative change in the progress function for each query. Ambainis~\textit{et al.}~\cite{AMRR11} later showed that the multiplicative bound is always at least as strong as the additive one, and therefore also characterizes bounded-error quantum query complexity.
Moreover, the multiplicative adversary method can prove better lower bounds for small success probability than the additive adversary method, and this was used to prove a strong direct product theorem for quantum query complexity~\cite{LR12}.

\paragraph{Quantum state generation.} Even when we are only interested in the quantum query complexity of functions, it is useful to also consider state generation problems: in that case, instead of producing the output $f(x)$ on input $x$, the algorithm is required to prepare a quantum state $\ket{m_x}$. Since unitary transformations independent of $x$ may be applied without any query to $x$, a quantum state generation problem is completely defined by the Gram matrix $\Mtarget=\sum_{x, x'}\braket{m_{x'}}{m_x}\ketbra{x}{x'}$. In the special case of computing a function, $M$ is a Boolean matrix. Thus every algorithm can be seen as generating a Gram matrix $M$.
If the algorithm is allowed some error $\eps$, then the set of Gram matrices that are acceptable outputs for the algorithm can be bounded by a so-called output condition. Different output conditions have been used before, for example, the original adversary method~\cite{Amb02} was implicitly using a condition based on the $L_\infty$ norm, while the adversary method with negative weights in~\cite{HLS07} was implicitly using the factorization norm $\gamma_2$. Realizing that different output conditions could be combined with different (zero-error) lower bound methods was key to comparing the additive and multiplicative adversary methods in~\cite{AMRR11}. More recently, Lee and Roland~\cite{LR12} were able to characterize exactly the set of acceptable Gram matrices, hence providing an optimal output condition (see \autoref{claim:OptimalOutputCondition}), which allowed them to prove a strong direct product theorem for quantum query complexity. This also simplifies the study of lower bounds techniques as it implies that the bounded-error quantum query complexity of a problem can be studied by bounding the zero-error quantum query complexity of all Gram matrices that define valid output states for the problem. As a consequence it is sufficient to compare the zero-error bounds for two methods in order to compare them.

\paragraph{Our results.} In this article, we tackle the second limitation by giving an explicit reduction from the polynomial method to the multiplicative adversary method. In order to do so, we introduce yet another lower bound technique for quantum query complexity, which we call the extended polynomial method (\autoref{def:xpoly} and \autoref{thm:xpoly}) as it can be seen as an extension of the polynomial method to Gram matrices. As the degree of a Boolean function can be stated as the maximum index of its Fourier coefficients, that is, $\deg(f) = \max \{ \abs{S} : \langle\chi_S, f\rangle \neq 0\}$, we define the degree of a Gram matrix by the maximum index $k$ such that the Gram matrix has support on a Fourier vector $\ket{\chi_S}$ with $\abs{S}=k$, that is, $\deg(M)=\max \{ \abs{S} : \bra{\chi_S}M\ket{\chi_S}\neq 0\}$.

For Boolean functions, the polynomial and the extended polynomial bounds are equal in the zero-error case. However, for the approximate case, the extended polynomial method uses the tight output condition, and is therefore possibly stronger than the polynomial method (\autoref{thm:polynomial-new}).

We also compare the extended polynomial method to the multiplicative adversary method. More particularly, we show that in the limit $c\to\infty$, where $c$ is the maximum multiplicative change in the progress function for one query, the multiplicative bound tends to the extended polynomial method (\autoref{thm:xpoly-madv}). This proof is constructive, i.e., we give an explicit multiplicative adversary matrix for which we have the equality. It might come as a surprise that this matrix does not depend on the problem: it is the same adversary matrix for every function. Let us note that it was proved in~\cite{AMRR11} that the multiplicative bound is stronger than the additive bound in the limit $c\to 1$, that is, at the other end of the possible range for $c$. This new result therefore completes the picture of the relations between the different lower bound techniques in quantum query complexity (see \autoref{fig:relations}), and shows in particular that all these methods reduce to the multiplicative adversary method.

\begin{figure}
\begin{center}
\begin{tikzpicture}[punkt/.style={rectangle, rounded corners, shade, top color=white, bottom color=blue!50!black!20, draw=blue!40!black!60, very thick, minimum width=2.2cm, minimum height= .9cm},
linet/.style={very thick,draw=blue!40!black!60, shorten >=2pt, shorten <=2pt}]
\node[punkt] at(0,0) (madv) {$\madv_\eps(f)$};
\node[punkt] at (-2,-2.9) (advpm) {$\adv^\pm_\eps(f)$};
\node[punkt] at (2,-2.9) (advtilde) {$\xpoly_\eps(f)$};
\node[punkt] at (-2,-5.8) (adv) {$\adv_\eps(f)$};
\node[punkt] at (2,-5.8) (poly) {$\adeg_\eps(f)$};
\node at (0,-5.8) {$\gtrless$}; \node at (0,-5.4) {\ding{176}};

\draw[->,linet] (poly) -- node[midway, right]{\ding{175}} (advtilde) ;
\draw[->,linet] (advtilde) to [bend right=10] node[midway,right] {\ding{175}} (madv);
\draw[->,linet] (adv) -- node[midway, left] {\ding{172}} (advpm) ;
\draw[->,linet] (advpm) to [bend left=10]  node[midway,above,] {\ding{173}} (madv);
\draw[<-, very thick, dotted, draw=red!40!black!60,shorten >=2pt, shorten <=2pt, bend right=10] (advpm) to node[midway, right] {\ding{174}} (madv) ;
\draw[<-, very thick, dotted, draw=red!40!black!60,shorten >=2pt, shorten <=2pt] (advpm) to [bend left=10] node[midway, right] {\ding{174}} (poly) ;
\end{tikzpicture}
\end{center}
\caption{\label{fig:relations}Relations between the different methods to prove lower bounds for quantum query complexity. An arrow from method $A$ to method $B$ implies that any lower bound that can be proved with $A$ can also be proved with $B$ (i.e., $B$ is stronger than $A$). A solid blue arrow means that the reduction is constructive, i.e., we can obtain a witness for $B$ from a witness for $A$. \ding{172} \cite{HLS07} \ding{173} \cite{AMRR11} \ding{174} \cite{Rei11,LMRSS11} \ding{175} [This article] \ding{176} The original additive and the polynomial methods are incomparable \cite{Zha05,SS06,AS04,Amb06}}
\end{figure}
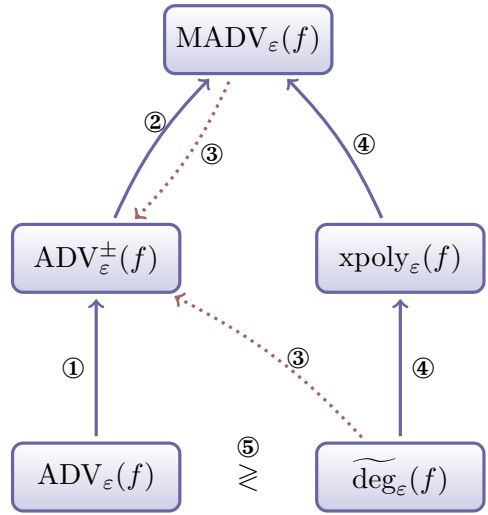

\section{Preliminaries}

\subsection{Gram matrices and fidelity}

\begin{definition}[Density matrices and Gram matrices]
 A \dt{density matrix} $\rho$ is a positive semidefinite matrix $\rho\succeq 0$ such that $\tr(\rho) = 1$. A \dt{normalized Gram matrix} $A$ is a positive semidefinite matrix $A\succeq 0$ such that $A\circ\I=\I$, where $\circ$ denotes the Hadamard (entry-wise) product.
\end{definition}
Note that any positive semidefinite matrix $A$ can be written as a Gram matrix in the broader sense, i.e., there always exists a set of vectors $\{\ket{a_x}\}$ such that $A_{xy}=\braket{a_x}{a_y}$. Here, the additional constraint $A\circ\I=\I$ means that we require those vectors to have norm $1$. Since all Gram matrices will be normalized in what follows, we will from now on refer to normalized Gram matrices as simply \emph{Gram matrices}.

\begin{definition}[Fidelity, Hadamard product fidelity]
	The \dt{fidelity} $\FF(\rho,\sigma)$ between two density matrices $\rho$ and $\sigma$ is defined by:
	\begin{align*}
		\FF(\rho,\sigma) = \tr\sqrt{\sqrt{\rho}\ \sigma \sqrt{\rho}}.
	\end{align*}
	The \dt{Hadamard product fidelity} $\HF(A,B)$ between two Gram matrices $A$ and $B$ is defined by:
	\begin{align*}
		\HF(A,B) = \min_{\ket u: \norm{\ket u}=1} \FF(A\circ \proj{u}, B \circ \proj{u}).
	\end{align*}
\end{definition}

The notation $\HF$ and the name Hadamard product fidelity\footnote{The name is chosen by analogy to the Hadamard product trace norm $\gamma_2$ (equivalent to the Hadamard product operator norm and also called factorization norm), which for Hermitian matrices can be written in the very similar form $\gamma_2(A)=\max_{\ket{u}:\norm{\ket{u}}\leq 1}\norm{A\circ \proj{u}}_{\tr}$.} are new to this article, but this quantity has been proved to be the tight output condition for the quantum query complexity in~\cite{LR12} (see \autoref{claim:OptimalOutputCondition} below).

\subsection{Quantum query complexity}

Consider a Boolean function $f:\{0,1\}^n\to\{0,1\}$. In the \emph{black-box model}, we are interested in computing $f(x)$ when $x$ is given by an oracle 
$O_x:\ket{i,b}\mapsto (-1)^{b \cdot x_i} \ket{i,b}$. We denote by $Q_\eps(f)$ the quantum query complexity of $f$, i.e., the minimum number of queries to $O_x$ necessary for any algorithm to output $f(x)$ with error at most $\eps$ (see, e.g.,~\cite{BuhrmanDeWolf02querysurvey}). Note that our choice of oracle computes the bits of $x$ in the phase. Another variant of this model considers an oracle that computes the bits in a register, but it can be shown that these models are equivalent.

Even when we are only interested in the quantum query complexity of functions, it is useful to also consider state generation problems~\cite{AMRR11,LMRSS11}. In that case, instead of producing the output $f(x)$ on input $x$, the algorithm is required to prepare a quantum state $\ket{m_x}\in\HH$. Since unitary transformations independent of $x$ may be applied without any query to $x$, a quantum state generation problem is completely defined by the Gram matrix $\Mtarget=\sum_{x, x'}\braket{m_{x'}}{m_x}\ketbra{x}{x'}$. For a quantum state generation problem specified by a Gram matrix $\Mtarget$, we define two different notions of query complexity. The coherent query complexity $\cqq_\eps(\Mtarget)$ is the minimum number of queries to the register oracle $O_x$ necessary to generate a state $\ket{n_x}\in\HH\otimes\HH'$ such that $\Re(\bra{n_x}(\ket{m_x}\otimes\ket{\bar{0}}))\geq\sqrt{1-\eps}$, where $\HH'$ is the workspace of the algorithm, $\ket{\bar{0}}\in\HH'$ is a default state for this workspace and $\Re(z)$ denotes the real part of a complex number $z$.  The non-coherent query complexity $\ncqq_\eps(\Mtarget)$ is defined similarly, except that it is enough to prepare a state $\ket{n_x}\in\HH\otimes\HH'$ such that $\Re(\bra{n_x}(\ket{m_x}\otimes\ket{m'_x}))\geq\sqrt{1-\eps}$, for an arbitrary set of states $\ket{m'_x}\in\HH'$ (that is, the workspace does not have to be reset to its default state).

For a Boolean function $f$, let us define the $\{1,-1\}$-valued function $\varphi: \{0,1\}^n \to \{1,-1\}: x \mapsto (-1)^{f(x)}$.
There are two natural quantum state generation problems associated to $f$, corresponding to the Gram matrices $F=\sum_{x,x'}\delta_{f(x),f(x')}\ketbra{x}{x'}$ and $\Phi=\sum_{x,x'}\varphi(x)\varphi(x')\ketbra{x'}{x}$, where $\delta$ is the Kronecker delta. Indeed, generating the Gram matrix $F$ non-coherently is exactly the same problem as computing $f$, and we therefore have $Q_\eps(f)=\ncqq_\eps(F)$, while generating the Gram matrix $\Phi$ coherently corresponds to \emph{computing the function in the phase}, i.e., we need to generate the state $\varphi(x)\ket{\bar{0}}$. The bounded-error complexities of these problems are closely related:
\begin{claim}[\cite{LR12}]
	\label{thm:qqch-phase-qqc-register}
	$Q_{(1-\sqrt{1-\eps})/2+\eps/4}(f)\leq \cqq_\eps(\Phi) \leq 2Q_{(1-\sqrt{1-\eps})/2}(f)$.
\end{claim}
This implies that to prove bounds on the bounded-error query complexity of $f$, it is sufficient to prove bounds on the query complexity of the related quantum state generation problem $\Phi$, and this is precisely the approach that we will use in this article.

Another advantage of considering quantum state generation problems is that we can study the bounded-error query complexity of a problem by bounding the zero-error query complexity of all Gram matrices that define valid output states for the problem. It follows from the following claim that this set of valid Gram matrices is characterized by the Hadamard product fidelity:
\begin{claim}[\cite{LR12}]
\label{claim:OptimalOutputCondition}
For any Gram matrix $M$ and any $\eps\geq 0$, we have
\begin{align*}
 Q_\eps(M)=\min_N\{Q_0(N):\HF(N,M)\geq\sqrt{1-\eps},\ N\succeq 0,\ N\circ\I=\I\}.
\end{align*}
\end{claim}

\subsection{The polynomial method}

\begin{definition}[Approximate degree]
For any $\eps\geq 0$, the \dt{approximate degree} $\adeg_\eps(f)$ of a function $f:\{0,1\}^n\to\R$ is defined as:
	\begin{align*}
		\adeg_\eps(f) &= \min_p \left\{ \deg(p) : \forall x \in\{0,1\}^n,\ \abs{p(x) - f(x)} \leq \eps \right\},
	\end{align*}
where the minimum is over $n$-variate polynomials $p:\R^n\to\R$.
\end{definition}


\begin{thm}[Polynomial method \cite{BBC+01}]
	If $f$ is a Boolean function, then $Q_\eps(f) \geq \Omega\!\left( \adeg_\eps(f) \right)$.
\end{thm}

In this article, we will use some basic Fourier analysis to relate degree of a function with Gram matrices. For the sake of readability, we will identify a set $S\subseteq \{1,\dots,n\}$ with its characteristic vector $S\in\{0,1\}^n$: $S_i = 1$ if and only if $i\in S$, and thus $\abs{S}$ can be either the cardinal of the set $S$ or the Hamming weight of the vector $S$.
\begin{definition}[Fourier basis and Fourier coefficients]
For any $S\in\{0,1\}^n$, let us define $\ket{\chi_S} = \frac{1}{\sqrt{2^n}}\sum_x (-1)^{S\cdot x} \ket{x}$. For a function $\varphi:\{0,1\}^n \to \R$, define the (non-normalized) state $\ket\varphi = \frac{1}{\sqrt{2^n}}\sum_x \varphi(x)\ket{x}$. We define the $S$-th \dt{Fourier coefficient} of $\varphi$ as $\hat\varphi(S)=\braket{\chi_S}{\varphi}$.
\end{definition}
Let us note that the set $\{\ket{\chi_S}\}_S$ is an orthonormal basis and that by definition, we then have $\hat\varphi(S)=\frac{1}{2^n}\sum_x (-1)^{S.x}\varphi(x)$ and $\varphi(x)=\sum_S (-1)^{S.x}\hat\varphi(S)$, which are the usual Fourier transform over the hypercube and its inverse. With these notations, we can also write the degree of a function $\varphi$ as $\deg(\varphi) = \max_S \{\abs{S}:\hat\varphi(S) \neq 0\}$.

\subsection{The multiplicative adversary method}
Let us consider a quantum algorithm generating the Gram matrix $M$ with error at most $\eps$ using $T$ queries. Let $\ket{\psi^t_x}$ be the state of the algorithm right after the $t$-th query when the input is $x$, and
$M^t=\sum_{x,x'}\braket{\psi^t_{x'}}{\psi^t_x}\ketbra{x}{x'}$ be the corresponding Gram matrix. Note that $M^0=\J$ and $M^T\approx \Mtarget$ (more precisely $\HF(M^T,M)\geq\sqrt{1-\eps}$).
The basic idea of all adversary methods is to design a Hermitian matrix $W$ defining a progress function $W[M]=\tr[WM]$ such that the initial value $W[\J]$ is low and the final value $W[M^T]$ is high (or vice versa), and then to bound the maximal change in the progress function for any oracle call. Whereas the additive method bounds the difference $|W[M^{t+1}]-W[M^t]|$, the multiplicative method bounds the ratio $W[M^{t+1}]/W[M^t]$. In this paper we use the definition of the multiplicative adversary method given by \cite{LR12} which is a slight extension of the original multiplicative adversary method in \cite{Spa08}.
\begin{definition}[Multiplicative adversary bound]
\label{dfn:multiplicative}
Let $\Mtarget$ be a Gram matrix specifying a quantum state generation problem and for all $i\in\{1,\cdots,n\}$, $D_i = \sum_{x,x'} (-1)^{x_i + x'_i} \ketbra{x}{x'}$ the action of the phase oracle on input $i$.
Fix $c>1$.
The \dt{multiplicative adversary bounds} are:
	\begin{align*}
		\madv_0^c (\Mtarget) &= \inv{\log c} \max_{W\succeq 0}\left\{ \log \tr [W \Mtarget]: \tr [W\J]=1,\ W\circ D_i \preceq c W\ \forall i\right\},\\
		\madv^c_\eps (\Mtarget) &=\min_N \left\{\madv^c_0(N) : \HF(N,M) \geq \sqrt{1-\eps},\ N \succeq 0,\ N\circ\I=\I\right\}, \\
		\madv_\eps (\Mtarget) &= \sup_{c>1}\madv_\eps^c (\Mtarget).
	\end{align*}
We call \dt{adversary matrix} for $\madv_0^c (\Mtarget)$ any matrix $W \succeq 0$ such that $\tr [W\J]=1$ and $W\circ D_i \preceq c W$ for all i.
\end{definition}

\paragraph{Remark.} Let us note that the parameter $c$ represents the maximum multiplicative change in the progress function that can result from one query. Since, for any matrix $W\succ 0$, the constraint $W\circ D_i \preceq c W$ is always satisfied for $c\geq\norm{(W\circ D_i)^{1/2} W^{-1/2}}^2$, one could directly obtain the multiplicative bound $\madv_0$ by optimizing over $W$ and taking $c=\norm{(W\circ D_i)^{1/2} W^{-1/2}}^2$. However, it is useful to define the bound $\madv_0^c$ for fixed $c$ as this can be expressed as a semidefinite program (see~\cite{LR12}), where the objective value is optimized over $W$. The best bound on the quantum query complexity is then obtained by maximizing the objective value over both $W$ and $c$.

\begin{thm}[Multiplicative adversary~\cite{Spa08,LR12}] For any $\eps \geq 0$ and any Gram matrix $\Mtarget$, we have 
		$\cqq_{\eps}(\Mtarget) \geq \madv_\eps (\Mtarget)$.
\end{thm}

\section{The extended polynomial method}


We now extend the polynomial method from Boolean functions to Gram matrices.

\begin{definition}[Extended polynomial bounds]
\label{def:xpoly}
Let $M$ be a Gram matrix specifying a quantum state generation problem. The \dt{extended polynomial bounds} are
\begin{align*}
	\xpoly_0(M) &= \max_S \{|S| : \tr\left[\proj{\chi_S}M\right] \neq 0\}, \\
	\xpoly_\eps(M) &= \min_N \left\{\xpoly_0(N) : \HF(N,M) \geq \sqrt{1-\eps},\ N\succeq 0,\ N\circ\I=\I \right\}.
\end{align*}

\end{definition}
\begin{thm}[Extended polynomial method] \label{thm:xpoly}
 For any $\eps \geq 0$ and any Gram matrix $\Mtarget$, we have 
	$\cqq_\eps(\Mtarget) \geq \xpoly_\eps(\Mtarget).$
\end{thm}
\begin{proof}
We prove the statement for $\eps=0$ and the general case immediately follows from \autoref{claim:OptimalOutputCondition} and the definition of $\xpoly_\eps(\Mtarget)$.
This proof actually considers the extended polynomial method as an adversary method. Let us define the progress function
\begin{align*}
	W[M^{t}] &= \max_{S} \left\{|S| : \tr[\proj{\chi_S}M^{t}] \neq 0 \right\}.
\end{align*}
Since $M^0 = \J = 2^n \proj{\chi_{\emptyset}}$, its initial value is $W[M^0] = 0$. The final value is $W[M^T]=\xpoly_0(M)$. It suffices to show that one query increases the progress function by at most one.

Let $M^t=\sum_i M_i^t$ be the Gram matrix just before the ($t+1$)-th query, where $M_i^t$ is the reduced Gram matrix corresponding to the part of the state where bit $x_i$ is queried (see, e.g.,~\cite{AMRR11} for details). Let $k=W[M^t]$ and note that by positivity, we have $\tr[\proj{\chi_S}M^{t}]= 0$ if and only if $\tr[\proj{\chi_S}M_i^{t}]= 0$ for all $i$. Therefore, we also have $W[M_i^t] \leq k$ for any $i$.

After the query, the Gram matrix of the algorithm will be $M^{t+1} = \sum_i M_i^t \circ D_i$. Let us observe that for any matrix $A$, we have $A\circ D_i = U_i A U^\dagger_i$ where $U_i = U^\dagger_i$ is the unitary matrix $U_i = \sum_x (-1)^{x_i} \proj{x}$. In particular, $\proj{\chi_S} \circ D_i = \proj{\chi_{S'}}$ where $S'= S\cup\{i\}$ if $i\not\in S$ and $S'= S\setminus\{i\}$ if $i\in S$.

 For all $S\in\{0,1\}^n$, we get:
\begin{align*}
	\tr\left[\proj{\chi_S}(M_i^t \circ D_i)\right] = \tr\left[(\proj{\chi_S}\circ D_i)M_i^t\right] = \!\sum_{T: |T|\leq k}\! \tr\left[(\proj{\chi_S}\circ D_i) \proj{\chi_T} M_i^t\right].
\end{align*}
This quantity is null for all $S$ such that $\abs{S} > k + 1$, therefore the progress function can increase by at most one per query.
\end{proof}


We have defined the extended polynomial method with the Fourier basis, but one might wonder if choosing another basis could provide better bounds. It turns out that this is not the case (the proof of this claim is deferred until \autoref{app:ProofMaxAdv}).
\begin{restatable}{claim}{claimmaxadv}
\label{thm:MaxAdv}
Let $\{\Pi_k:0\leq k\leq K\}$ be a set of orthogonal projectors such that
\begin{enumerate}
	\item[\ding{172}] $\sum_{k} \Pi_{k} = \I_{\C^{2^n}}$, 
	\item[\ding{173}] $\tr(\Pi_0\J)= 2^n$,
	\item[\ding{174}] $\forall i\in\{1,\dots,n\},\ \forall l, k$ such that $|l-k|>1,\ \tr\!\left[(\Pi_{l}\circ D_i)\Pi_{k}\right] = 0$.
\end{enumerate}
Then, for any Gram matrix $\Mtarget$, we have
\begin{align*}
	Q_0(M)\geq\xpoly_0(M) \geq \max_{k} \left\{ k : \tr\!\left(\Pi_{k}\Mtarget\right)\neq 0\right\}.
\end{align*}
\end{restatable}
Therefore, while any set of projectors provides a lower bound on quantum query complexity, the best bound is achieved by the extended polynomial method, which corresponds to the special case $K=n$ and $\Pi_k = \sum_{S:|S|=k} \proj{\chi_S}$.

\section{Relation between the polynomial and the extended polynomial methods}

In this Section, we compare the strength of the polynomial and the extended polynomial methods. Let $f$ be a Boolean function and $\Phi$ the Gram matrix corresponding to computing $f$ in the phase. By definition of the extended polynomial method, we have that $\xpoly_0(\Phi) = \deg(f)$. However the equality is lost in the approximate case:
\begin{thm}\label{thm:polynomial-new}
	Let $f$ be a Boolean function and $\Phi$ be the Gram matrix corresponding to computing $f$ in the phase. For any $\eps \geq 0$, we have
	\begin{align*}
		\xpoly_{\eps}(\Phi) \geq \adeg_{\eps/2}(f).
	\end{align*}
\end{thm}

\begin{proof}
We first show that $\xpoly$ can be written as an optimization problem over polynomials. By definition, we have
\begin{align}\label{eq:maxadv-eps-error}
 \xpoly_\eps(\Phi)=\min_N\xpoly_0(N),
\end{align}
where the minimum is taken over positive semidefinite matrices $N$ such that $N\circ \I=\I$ and $\HF(N,\Phi)\geq\sqrt{1-\eps}$.

Let us write $\Phi$ as a Gram matrix $\Phi=\sum_{x,y}\braket{\varphi_x}{\varphi_y}\ketbra{y}{x}$, where $\ket{\varphi_x}=(-1)^{f(x)}\ket{0}$. Then, by the properties of the Hadamard fidelity~\cite{LR12}, the minimum in \autoref{eq:maxadv-eps-error} can be taken over Gram matrices $N=\sum_{x,y}\braket{\psi_x}{\psi_y}\ketbra{y}{x}$ such that $\ket{\psi_x}$ is a unit vector for any $x$ and $\Re(\braket{\varphi_x}{\psi_x})\geq\sqrt{1-\eps}$ for any $x$. Writing any unit vector $\ket{\psi_x}$ in the computational basis as $\ket{\psi_x}=\sum_i p_i(x)\ket{i}$, the minimum in \autoref{eq:maxadv-eps-error} can equivalently be taken over amplitudes $p_i(x)$, therefore
\begin{align*}
  \xpoly_\eps(\Phi)=\min_{(p_i)}\xpoly_0\left(\sum_{x,y,i} p_i^*(x)\; p_i(y)\ketbra{y}{x}\right),
\end{align*}
where the minimum is taken over complex functions $p_i:\{0,1\}\to\C$ such that $\sum_i |p_i(x)|^2=1$ and $(-1)^{f(x)}\Re(p_0(x))\geq\sqrt{1-\eps}$.

Finally, we show that for $N=\sum_{x,y,i} p_i^*(x) p_i(y)\ketbra{y}{x}$, we have $\xpoly_0(N)=\max_i(\deg(p_i))$. For all $S\subseteq \{1,\dots,n\}$, we have
\begin{align*}
\bra{\chi_S}N\ket{\chi_S} =\sum_{i,x,y} p_i^*(x) p_i(y)\braket{\chi_S}{y}\braket{x}{\chi_S} = 2^n \sum_i|\hat{p}_i(S)|^2,
\end{align*}
where $\hat{p}_i(S)$ are the Fourier coefficients of $p_i$. This is nonzero only if there exists an $i$ such that $\deg(p_i)\geq k$, hence $\xpoly_0(N)=\max_i(\deg(p_i))$.

To summarize, this implies that for any Boolean function $f$ with associated phase matrix $\Phi$, we have
\begin{align*}
   \xpoly_\eps(\Phi)=\min_{(p_i)}\max_i(\deg(p_i)),
\end{align*}
where the minimum is taken over a set of functions $p_i:\{0,1\}^n\to\R$ satisfying
\begin{enumerate}
 \item $\sum_i p_i(x)^2=1$ for any $x\in\{0,1\}^n$,
 \item $(-1)^{f(x)}p_0(x)\geq\sqrt{1-\eps}$ for any $x\in\{0,1\}^n$.
\end{enumerate}
Optimizing over real polynomials instead of complex ones is without loss of generality, since any complex polynomial $p_j$ can be replaced by two polynomials being the real and the imaginary part of $p_j$, and the condition $\sum_i \abs{p_i(x)}^2 = 1$ would still be satisfied.

Let $\{p_i\}$ be a set of polynomials such that $\xpoly_\eps(\Phi)=\max_i(\deg(p_i))$ and satisfying the required conditions. In particular, we have $\sqrt{1-\eps}\leq (-1)^{f(x)}p_0(x)\leq 1$ and $\deg(p_0)\leq\xpoly_\eps(\Phi)$. Setting $p(x)=(1-p_0(x))/2$, we have $\deg(p)=\deg(p_0)\leq\xpoly_\eps(\Phi)$ and
\begin{align*}
|p(x)-f(x)|\leq \frac{1-\sqrt{1-\eps}}{2}\leq\frac{\eps}{2},
\end{align*}
for any $x$, so that $p$ witnesses that $\adeg_{\eps/2}(f)\leq \deg(p)$.
%
%
\end{proof}

\section{Relation with the multiplicative adversary method}
In~\cite{AMRR11}, it was shown that in the limit $c\to 1$, the multiplicative adversary bound $\madv_0^c(\Mtarget)$ is at least as strong as the additive adversary bound $\adv^\pm(\Mtarget)$. Here, we show that the extended polynomial bound can be obtained by taking the limit $c\to\infty$.

\begin{thm}
\label{thm:xpoly-madv}
Let $\Mtarget$ be a Gram matrix, $\eps\geq 0$, $T=\xpoly_\eps(\Mtarget)$ and $\Pi_{\geq T}=\sum_{S:\abs{S}\geq T}\proj{\chi_S}$. Moreover, let $\delta>0$ be such that $\tr[\Pi_{\geq T} N]\geq\delta$ for any Gram matrix $N$ such that $\HF(N,\Mtarget) \geq \sqrt{1-\eps}$.

Then, for any $c>1$, we have
\begin{align*}
\xpoly_\eps(\Mtarget)-\frac{n-\log\delta}{\log c}
\leq
\madv_\eps^c(\Mtarget)
\leq
\xpoly_\eps(\Mtarget)+\frac{n}{\log c}.
\end{align*}
In particular, in the limit $c\to\infty$, we have
\begin{align*}
 \lim_{c\to\infty}\madv_\eps^c(\Mtarget)=\xpoly_\eps(\Mtarget).
\end{align*}
\end{thm}

\paragraph{Remark.} Note that such a value of $\delta$ always exists. The quantity $\tr[\Pi_{\geq T}N]$ is non negative as it is the trace of the product of semidefinite matrices, and it cannot be equal to $0$. Assume by contradiction that $\tr[\Pi_{\geq T}N]=0$, then $\xpoly_0(N) \leq T-1$, however $N$ is an $\eps$-approximation of $M$ that has a polynomial bound of $T$.

The general idea of the proof is to consider the multiplicative adversary matrix 
\begin{align*}
	W =  \frac{1}{2^n}\sum_S c^{\abs{S}}\proj{\chi_S}
\end{align*}
as a multiplicative adversary matrix. The lower bound then follows from the fact that in the limit $c\to\infty$, the value of the progress function $W[\Mtarget]=\tr[W\Mtarget]$ will be dominated by the term in $c^{\abs{S}}$ for the set $S$ with the largest size $\abs{S}=k$ such that $\bra{\chi_S}\Mtarget\ket{\chi_S}\neq 0$, which therefore corresponds to the degree of the matrix $\Mtarget$. As for the upper bound, we show that the matrix $W$ becomes an optimal multiplicative adversary matrix in the limit $c\to\infty$. This can be shown by observing that one oracle call can only map a Fourier basis state $\ket{\chi_S}$ to another Fourier basis state $\ket{\chi_{S'}}$ with $\abs{S'}=\abs{S}\pm 1$ which implies bounds on the elements of any possible multiplicative adversary matrix written in the Fourier basis.


%
\begin{proof}
	We prove it for the zero-error case, the general case follows immediately. 

	Consider the matrix $W =  \frac{1}{2^n}\sum_S c^{|S|}\proj{\chi_S}$. It is a valid adversary matrix for $\madv^c_0(M)$ since $\tr[W\J] = 1$ and $W\circ D_i \preceq cW,\ \forall i\in\{1,\dots,n\}$. This inequality follows from $W\circ D_i = \frac{1}{2^n}\left(\sum_{S:i\in S}c^{\abs{S}-1}\proj{\chi_S}+\sum_{S:i\not\in S} c^{\abs{S}+1}\proj{\chi_S}\right)$, see proof of \autoref{thm:xpoly}. Let $W'$ be an optimal multiplicative adversary matrix for $\madv^c_0(M)$. Let us show that $\tr[WM] \leq \tr[W'M] \leq 2^n \tr[WM]$.

The first inequality is a direct consequence of the fact that $W$ is an adversary matrix for $\madv^c_0(M)$ and the definition of the multiplicative adversary bound. 

To prove the second inequality, let us first show by induction on $k=\abs{S}$ that $\bra{\chi_S}W'\ket{\chi_S} \leq \frac{1}{2^n} c^{|S|}$ for any set $S$. 
	For $k=0$, the condition $\tr[W'\J] = 1$ is equivalent to $\bra{\chi_\emptyset}W'\ket{\chi_\emptyset} = \frac{1}{2^n}$. 
	
	Let us fix $0\leq k \leq n$, and assume that $\forall S$ such that $|S|=k$, we have $\bra{\chi_S}W'\ket{\chi_S} \leq \frac{1}{2^n}c^k$. Let $S'$ be a set of size $k+1$ and decompose it into $S'= S \cup \{i\}$. Observe first that $\bra{\chi_S}W'\circ D_i\ket{\chi_S} = \bra{\chi_S}U_i W' U_i \ket{\chi_S} = \bra{\chi_{S'}}W' \ket{\chi_{S'}}$ where $U_i = \sum_x (-1)^{x_i} \proj{x}$ as defined in the proof of \autoref{thm:xpoly}. Hence by sandwiching $W'\circ D_i \preceq cW'$ with $\ket{\chi_S}$, we get $\bra{\chi_{S'}}W' \ket{\chi_{S'}} \leq c \bra{\chi_{S}}W' \ket{\chi_{S}}  \leq \frac{1}{2^n}c^{|S|+1}$.
	
	We can now proceed with the rest of the proof:
\begin{align*}
	\tr[W'M] &= \sum_S \bra{\chi_{S}}W' M \ket{\chi_{S}} 
	 = \sum_{S,S'}\bra{\chi_{S}}W' \ket{\chi_{S'}}\!\bra{\chi_{S'}} M \ket{\chi_{S}} \\
	& \leq  \sum_{S,S'} \abs{\bra{\chi_{S}}W' \ket{\chi_{S'}}} \abs{\bra{\chi_{S'}} M \ket{\chi_{S}}}.
\end{align*}
We now use the property that for any positive semidefinite matrix $A$, $\abs{A_{ij}} \leq \sqrt{A_{ii}A_{jj}}$,
\begin{align*}
	\tr[W'M] 
	& \leq \left( \sum_S \sqrt{\bra{\chi_{S}}W' \ket{\chi_S}\bra{\chi_S} M \ket{\chi_{S} }}\right)^2.
\end{align*}
Using the Cauchy-Schwarz inequality, we get:
\begin{align*}
\tr[W'M]
	 \leq 2^n \sum_S \bra{\chi_{S}}W' \ket{\chi_S}\bra{\chi_S} M \ket{\chi_{S} } 
	 \leq  \sum_S c^{|S|} \bra{\chi_S} M \ket{\chi_S} = 2^n \tr[WM].
\end{align*}

We are now ready to conclude the proof. From $\tr[WM] \leq \tr[W'M] \leq 2^n \tr[WM]$, we have by definition of $\madv^c_0(M)$
\begin{align*}
 \frac{\log\tr[WM]}{\log c}\leq\madv^c_0(M)\leq\frac{n+\log\tr[WM]}{\log c}.
\end{align*}

For $T=\xpoly_\eps(\Mtarget)$, we find from the first inequality
\begin{align*}
	\madv^c_0(\Mtarget) \geq \frac{\log \frac{1}{2^n} c^T \tr[\Pi_{\geq T} M] }{\log c} = T + \frac{\log(\tr[\Pi_{\geq T} M]) - n}{\log c}.
\end{align*}

Similarly, from the second inequality, we have
\begin{align*}
\madv^c_0(M) \leq \frac{\log \sum_S c^{\abs{S}}\bra{\chi_S} M \ket{\chi_{S} }}{\log c} \leq T + \frac{\log \sum_S \bra{\chi_S} M \ket{\chi_{S} } }{\log c}=T+\frac{n}{\log c},
\end{align*}
where we used the facts that $\bra{\chi_S} M \ket{\chi_{S}}=0$ whenever $\abs{S}>T$, and $\sum_S \bra{\chi_S} M \ket{\chi_{S}} = \tr[M]=2^n$.
\end{proof}

We note that $\madv^c_\eps(\Mtarget)$ approaches its limiting value $\xpoly_\eps(\Mtarget)$ if $c$ is large enough compared to $2^n/\delta$. In general, we cannot give a lower bound on $\delta$ in order to determine how large $c$ should be. However, for the special case of Boolean functions, and comparing to the standard polynomial method, i.e., the approximate degree $\adeg_\eps(f)$, instead of $\xpoly_\eps(\Mtarget)$, we can deduce such a general bound on how large $c$ should be, based on the following fact:
\begin{restatable}{fact}{factlargeoverlap}\label{fact:large-overlap}
Let $f$ be a Boolean function with approximate degree $T=\adeg_\eps(f)$ and $p$ be a polynomial such that $\sqrt{1-\eps}\leq (-1)^{f(x)}p(x)\leq 1$ for any $x$. Then, we have
\begin{align*}
 \sum_{S:\abs{S}\geq T} \abs{\hat{p}(S)}^2 \geq \frac{\eps^2}{2^n}.
\end{align*}
\end{restatable}
We present here a proof by contradiction, but there is also an alternative proof using dual polynomials proposed to us by \v{S}palek and reproduced in \autoref{app:AlternativeProof}.
\begin{proof}
Assume towards contradiction that $\sum_{S:\abs{S}\geq T} \abs{\hat{p}(S)}^2 < \frac{\eps^2}{2^n}$, and let $q(x)=(1-p(x))/2$. Then, we have $\abs{q(x)-f(x)}\leq\eps/2$ for all $x$ and $\sum_{S:\abs{S}\geq T} \abs{\hat{q}(S)}^2< \frac{\eps^2}{4\cdot 2^n}$ by assumption on $p$. Let $q'(x)=\sum_{S:\abs{S}<T}(-1)^{S\cdot x}\hat{q}(S)$, so that $\deg(q')<T$ and
\begin{align*}
 \abs{q'(x)-q(x)}^2=\abs{\sum_{S:\abs{S}\geq T}(-1)^{S\cdot x}\hat{q}(S)}^2 
\leq 2^n\sum_{S:\abs{S}\geq T} \abs{\hat{q}(S)}^2<\frac{\eps^2}{4}
\end{align*}
for all $x$, where we have used the Cauchy-Schwarz inequality. Therefore, $\abs{q'(x)-f(x)}<\eps$ for all $x$ and the polynomial $q'$ witnesses that $\adeg_\eps(f)\leq\deg(q')<T$, a contradiction.
\end{proof}

This fact implies that ${\madv}^c_\eps(\Phi)$ becomes at least as strong as $\adeg_\eps(f)$ as soon as $c$ is large compared to $2^n/\eps$.
\begin{lem}
\label{thm:MADVApproximateDegree}
  Let $f$ be a Boolean function with associated phase matrix $\Phi$.
Then, for any $c>1$, we have
\begin{align*}
   {\madv}^c_\eps(\Phi)\geq\adeg_\eps(f)-2\cdot\frac{n-\log\eps}{\log c}.
\end{align*}
\end{lem}

\begin{proof}[Proof (sketch)]
Let $W = \frac{1}{2^n} \sum_S c^{\abs{S}}\proj{\chi_S}$. By definition of the multiplicative adversary method, $\madv^c_\eps(\Phi) \geq \min_N \frac{\log\tr(WN)}{\log c}$ where the minimum is taken over all Gram matrices $N$ such that $\HF(N,M) \geq \sqrt{1-\eps}$. We follow the same steps as in the proof of \autoref{thm:polynomial-new}: we express the Gram matrix $N$ as $N=\sum_{x,y}\braket{\psi_x}{\psi_y}\ketbra{y}{x}$ and parametrize the states $\ket{\psi_x}$ as $\ket{\psi_x}=\sum_i p_i(x)\ket{i}$. After relaxing the normalization condition on the states $\ket{\psi_x}$, we obtain that
\begin{align*}
	\madv^c_\eps(\Phi)\geq\frac{1}{\log c}\log \min_{p}\frac{1}{2^n}\sum_S c^{\abs{S}}\abs{\hat{p}(S)}^2,
\end{align*}
where the minimum is taken over all polynomials $p:\{0,1\}^n \mapsto \R$ satisfying $\sqrt{1-\eps}\leq (-1)^{f(x)}p(x)\leq 1$ for any $x\in\{0,1\}^n$.

Let $p$ be a polynomial achieving this minimum and $T=\adeg_\eps(f)$. Then, we have
\begin{align*}
   {\madv}^c_\eps(\Phi)&\geq\frac{1}{\log c}\log\left( \frac{1}{2^n}\sum_S c^{\abs{S}}\abs{\hat{p}(S)}^2\right) 
 \geq T- \frac{n-\log\left(\sum_{S:\abs{S}\geq T} \abs{\hat{p}(S)}^2\right)}{\log c}.
\end{align*}
The lemma then follows from \autoref{fact:large-overlap}.
\end{proof}

\section{Discussion and open questions}

Strong connections have been known for quite some time between the approximate degree of a function and its query complexity: they are polynomially related for all (total) functions for classical complexity~\cite{NS94} as well as for quantum complexity~\cite{BBC+01}. The latter is actually often equal to the approximate degree (at least up to a constant factor) for many functions, including all symmetric functions and random functions. With a large number of tight bounds proved using the polynomial method~\cite{BBC+01,AS04,Amb05-collision,AdW12} to cite only a few, this method might even seem ubiquitous. However, it is not always tight as in some rare cases the adversary method is known to yield better bounds. By clarifying the relation between the polynomial method and adversary bounds, this work provides some new insight on why this can be the case.

First, we showed that the polynomial method is a relaxation of a more general method which we called the extended polynomial method. This has a particularly nice interpretation when one wants to compute the value of a function in a register, i.e., the goal is to prepare the state $\ket{f(x)}$.\footnote{This is the standard problem studied in most articles on quantum query complexity, even though some recent works including this one have considered the problem of computing the function in the phase. Recall that \autoref{claim:OptimalOutputCondition} implies that both problems are equivalent.}
When error $\eps$ is allowed, measuring this register should yield outcome $f(x)$ with probability at least $1-\eps$, that is, the probability $p(x)$ of obtaining outcome $1$ should be close to $1$ when $f(x)=1$ and close to $0$ when $f(x)=0$. While the polynomial method only considers the degree of the probability $p(x)$, the extended polynomial method considers the degree of all the amplitudes in the final state of the algorithm, including the erroneous part. In terms of Gram matrices this corresponds to relaxing the condition $N \circ \I = \I$ to $N \circ \I \preceq \I$.\footnote{Note that with the relaxed condition $N \circ \I \preceq \I$, the matrix $N$ does not have to be a \emph{normalized} Gram matrix anymore, in which case the Hadamard product fidelity is not defined. However, one can use another output condition, for example $\gamma_2(N-M)\leq\sqrt{2\eps}$, where $\gamma_2$ denotes the Hadamard product trace norm. These output conditions are related up to a constant~\cite{LMRSS11, LR12}, so that it only affects the lower bound by at most a constant factor for bounded-error query complexity.}

In general it is not known how large the gap between the polynomial and the extended polynomial method can be. It appears to be larger by at least a factor two for some functions. Indeed, Ambainis \textit{et al.} improved the lower bound for random Boolean functions from $n/4-o(n)$ using the polynomial method, to $n/2-o(n)$ (which is tight) by bounding the degree of all amplitudes in the final state of the algorithm~\cite{ABSdW13} (their argument can be seen as a special case of the extended polynomial method).

Secondly this provides a partial answer on how the multiplicative adversary method $\madv^c$ varies with $c$. Indeed, while it was already known that $\madv^{c\to 1}_\eps(f)\geq\adv^\pm_\eps(f)$, we have proved that $\madv^{c\to\infty}_\eps(f)\geq\adeg_\eps(f)$, and in particular, $\madv^{c\to\infty}_0(f)=\deg(f)$ in the zero-error case. This implies that the gap between $\madv$ and $\madv^{c\to\infty}$ can be at least polynomially large by considering the Ambainis function~\cite{Amb06}, for which the polynomial method fails to give a tight bound, contrary to the adversary method. This gap might be explained by the fact that in the limit $c\to\infty$, the eigenbasis of the best adversary matrix is restricted to be the Fourier basis, while for smaller values, other bases can provide better bounds.

To summarize our current knowledge, the situation is the following. On the one hand, when $c$ tends to one, the multiplicative adversary method is tight for bounded-error (\cite{AMRR11}) but not for zero-error (e.g., for the \pbsc{OR} function, there is a quadratic gap). On the other hand, when $c$ tends to infinity, the multiplicative method seems better for zero-error as it proves the $\Omega(n)$ lower bound for \pbsc{OR}, but it is not always tight (Ambainis function). As for low success probability, it seems that taking $c$ bounded away from one provides an advantage, as shown in particular by the strong direct product theorems proved using the multiplicative~\cite{Spa08,LR12} and polynomial methods~\cite{KSW07,She11}.

This leaves open a few interesting questions about the behavior of the multiplicative adversary method. Can we say more about the dependence of $\madv^c$ on $c$? Can we improve the relation $\madv^{c\to1}_\eps(M) \geq \textrm{ADV}^\pm_\eps(M)$ to an equality in general? Can we characterize the set of functions for which the (extended or not) polynomial method does not provide a tight bound? Finally, does the multiplicative adversary method characterize the quantum query complexity, i.e., is it tight for any error?


\section*{Acknowledgements}
Most of this work was done at NEC Laboratories America. The authors thank M. R\"otteler, D. Gavinsky, and T. Lee for stimulating discussions; and R. de~Wolf and R. \v{S}palek for interesting comments. They also thank R. \v{S}palek for proposing the alternative proof of \autoref{fact:large-overlap} using dual polynomials. This work was supported by ARO/NSA under grant W911NF-09-1-0569. L.M. also acknowledges the support of the Ministry of Education and the National Research Foundation, Singapore. J.R. also acknowledges support from the action \emph{Mandats de Retour} of the \emph{Politique Scientifique F\'ed\'erale Belge} and the Belgian ARC project COPHYMA.

\bibliography{polyvsadversary}

\appendix

\section{Proof of \autoref{thm:MaxAdv}}
\label{app:ProofMaxAdv}

\claimmaxadv*
\begin{proof}
Let $\left\{\Pi'_k:0\leq k\leq K\right\}$ be any set of projectors satisfying the three conditions in \autoref{thm:MaxAdv}, and let $\SS'_k$ be the subspace on which $\Pi'_k$ projects. Let $\SS_k=\Span\{\ket{\chi_S}:|S|=k\}$ be the subspace on which $\Pi_k = \sum_{S:\abs{S}=k} \proj{\chi_S}$ projects. Finally, let $\SS_{\leq k}=\bigoplus_{l=0}^k \SS_l$ and $\SS'_{\leq k}=\bigoplus_{l=0}^k \SS'_l$, and similarly for $\SS_{>k}$ and $\SS'_{>k}$, as well as for the corresponding projectors $\Pi_{\leq k},\Pi'_{\leq k},\Pi_{>k}$ and $\Pi'_{>k}$.

We are going to show that for all $0 \leq k \leq n$, we have $\SS_{\leq k} \subseteq \SS'_{\leq k}$, which directly concludes the proof.  We show it by induction on $k$.

For $k=0$, the property \ding{173} reads $\tr(\Pi'_0\proj{\chi_{\emptyset}})=1$, hence $\SS_0=\Span\{\ket{\chi_{\emptyset}}\}\subseteq \SS'_0$.

Let us now fix $0<k<n$ and assume that $\SS_{\leq k}\subseteq \SS'_{\leq k}$. Since $\SS_{\leq k}\subseteq \SS'_{\leq k}\subseteq \SS'_{\leq k+1}$, it is sufficient to prove that $\forall S:\abs{S}=k+1,\ \ket{\chi_S} \in \SS'_{\leq k+1}$. 

Fix $S = S' \cup \{i\}$ such that $\abs{S} = k+1$. By property \ding{174}, we have
\begin{align*}
  \tr\left[ (\Pi'_{l}\circ D_i)\Pi'_{m} \right] = 0,
\end{align*}
for any $l,m$ such that $|l-m|>1$. Summing this equation over $0\leq l\leq k$ and $k+2\leq m \leq K$, we obtain that
\begin{align*}
	\tr\left[ (\Pi'_{\leq k}\circ D_i)\Pi'_{>k+1} \right] = 0.
\end{align*}
Since, by assumption, we have $\SS_{\leq k}\subseteq \SS'_{\leq k}$, this implies that
\begin{align*}
  	0 \leq \tr\left[(\Pi_{\leq k}\circ D_i)\Pi'_{>k+1}\right]\leq  \tr\left[(\Pi'_{\leq k}\circ D_i)\Pi'_{>k+1}\right]= 0,
\end{align*}
We use the property that $\tr(AB) \geq 0$ if $A$ and $B$ are positive semidefinite to prove the first inequality. Since all the terms in the decomposition $\Pi_{\leq k} = \sum_T \proj{\chi_T}$ are positive semidefinite matrices, we have in particular that $\tr\left[\left(\proj{\chi_{S'}}\circ D_i \right)\Pi'_{>k+1}\right] = 0$, therefore,
\begin{align*}
  \tr\left[\proj{\chi_S}\Pi'_{>k+1}\right]= 0.
\end{align*}
Since $\Pi'_{>k+1}+\Pi'_{\leq k+1}=\I$ by property \ding{172}, we can conclude that $\tr\!\left[\proj{\chi_S}\Pi'_{\leq k+1}\right]= 1$, hence $\ket{\chi_S}\in \SS'_{\leq k+1}$.
\end{proof}

\section{Alternative proof of \autoref{fact:large-overlap} }
\label{app:AlternativeProof}
\factlargeoverlap*

In this appendix, we present an alternative proof of this fact, based on the notion of dual polynomial:
\begin{lem}
\label{lem:DualPoly}
 Let $f:\{0,1\}^n\to\{0,1\}$ be a function with approximate degree $T=\adeg_\eps(f)$. Then, there exists a polynomial $d:\{0,1\}^n\to\R$, called \dt{dual polynomial}, such that
\begin{enumerate}
 \item $\sum_x \abs{d(x)} = 1$,
 \item $\sum_{x}d(x)f(x)\geq\eps$,
 \item $\sum_x d(x)\braket{x}{\chi_S}=0$ for all $S$ such that $|S|<T$.
\end{enumerate}
\end{lem}

Let us now prove the fact.
\begin{proof}[Proof of \autoref{fact:large-overlap}]
Let $p$ be a polynomial achieving this minimum, and $d$ be a dual polynomial witnessing that $\adeg_\eps(f)=T$. Defining the polynomial $q$ as $q(x)=\frac{1-p(x)}{2}$, we have $\hat{p}(S)=2\hat{q}(S)$ whenever $\abs{S}\neq 0$, and $|q(x)-f(x)|\leq\eps/2$ by assumption on $p$. Moreover, by definition of $d$, we have $|\hat{d}(S)|=0$ for $|S|<T$ and $|\hat{d}(S)|\leq 1$ for any $S$, so that
\begin{align*}
	\sum_{S:\abs{S}\geq T} \abs{\hat{p}(S)}^2
	\geq 4 \sum_{S:\abs{S}\geq T} \abs{\hat{q}(S)}^2
	\geq 4 \sum_{S} \abs{\hat{d}(S)}^2\abs{\hat{q}(S)}^2.
\end{align*}
Using the Cauchy-Schwarz inequality, we get
\begin{align*}
	\sum_{S:\abs{S}\geq T} \abs{\hat{p}(S)}^2
	\geq\frac{4}{2^n}\abs{\sum_{S}\hat{d}(S)\hat{q}(S)}^2
	=\frac{4}{2^n}\abs{\sum_x d(x)q(x)}^2.
\end{align*}
Let us define $e(x)=2(q(x)-f(x))/\eps$ which is such that $\abs{e(x)}\leq 1$ for any $x$. Since $\sum_x \abs{d(x)} \leq 1$ and $\sum_x d(x) f(x)\geq\eps$, we have
\begin{align*}
	\sum_{S:\abs{S}\geq T} \abs{\hat{p}(S)}^2
	\geq \frac{4}{2^n}\abs{\sum_x d(x)f(x)+(\eps/2) \sum d(x)e(x)}^2
	\geq \frac{4}{2^n}\abs{\eps-\eps/2}^2=\frac{\eps^2}{2^n}.
\end{align*}
\end{proof}

\end{document}